\documentclass[12pt]{article}      
\usepackage{nicematrix,fourier}
\usepackage{array}
\usepackage{tikz-cd}
\newcolumntype{C}{>{$}c<{$}}

\usepackage[symbol]{footmisc}

\usepackage{amssymb}
\usepackage{eucal}
\usepackage{amsmath}
\usepackage{amsthm}
\usepackage{graphicx}
\usepackage{float}
\usepackage{framed}
\usepackage{enumerate}
\usepackage{hyperref}
\usepackage[toc,page]{appendix}

\newcommand{\curly}{\mathrel{\leadsto}}

\newcommand{\R}{{\mathbb  R}}   
 
\newcounter{mastercounter}
\numberwithin{mastercounter}{section}

\newtheorem{thm}[mastercounter]{Theorem}
\newtheorem{lemma}[mastercounter]{Lemma}
\newtheorem{prop}[mastercounter]{Proposition}

\makeatletter
\let\c@equation\c@mastercounter
\makeatother

\numberwithin{equation}{section}

\usepackage[T1]{fontenc}

\DeclareMathOperator{\tr}{tr}
\DeclareMathOperator{\hs}{Hess}
\DeclareMathOperator{\bu}{{\bf u}}

\begin{document}

\title{ Upper block triangular form for the Laplace-Beltrami operator on the special orthogonal group acquired through a flag of trace polynomials spaces
}

\author{Petre Birtea, Ioan Ca\c su, Dan Com\u{a}nescu\\
{\small Department of Mathematics, West University of Timi\c soara} 
\\
{\small Bd. V. P\^ arvan, No 4, 300223 Timi\c soara, Rom\^ania}\\
{\small Email: petre.birtea@e-uvt.ro, ioan.casu@e-uvt.ro, dan.comanescu@e-uvt.ro}}
\date{}

\maketitle

\begin{abstract}
The Laplacian of a general trace polynomial function defined on the special orthogonal group $SO(N)$ is explicitly computed. An invariant flag of spaces generated by trace polynomials is constructed. The matrix of the Laplace-Beltrami operator on $SO(N)$ for this flag of vector spaces takes an upper block triangular form. As a consequence of this construction, the eigenvalues and eigenfunctions of the Laplace-Beltrami operator on $SO(N)$ can be computed in an iterative manner. For the particular cases of the special orthogonal groups $SO(3)$ and $SO(4)$ the complete list of eigenvalues is obtained and the corresponding irreducible characters of the representation for these groups are expressed as trace polynomials.

\end{abstract}

\maketitle

\section{Introduction}

At the core of the Fourier analysis stands the spectral theory of Laplace-Beltrami operator on the circle $S^1$ regarded as a Lie group. Since then a lot of effort has been dedicated to generalize the concepts and results of Fourier analysis to more complicated Lie groups. One of the most powerful tools  is the link between the representation theory of Lie groups and the spectral theory of Laplace-Beltrami operator on such Lie groups.

In this paper we take a more computational (formula oriented) approach to the study of Laplace-Beltrami operator on special orthogonal group $SO(N)$. More precisely, we construct a flag of invariant polynomial subspaces on which this operator takes upper block triangular form. This particularity of the construction leads to an iterative  way to compute the eigenvalues. Another outcome of the upper block triangular form is the possibility to write the irreducible characters of special orthogonal groups as trace polynomials. The detailed structure of the paper is as follows.

In Section \ref{spherical-eigenvalues-121}, using the explicit formula for the Laplacian $\Delta_{SO(N)}$ on the special orthogonal group introduced in \cite{laplacian},  we give a computational proof of the classical result which states that the eigenvalues of the Laplace-Beltrami operator on the sphere $S^{N-1}$, divided by 2, i.e. $-\frac{k(k+N-2)}{2}$, $k\in \mathbb{N}$, are among the eigenvalues of the Laplace-Beltrami operator $\Delta_{SO(N)}$.  Also, we prove that the Gegenbauer polynomials $C_k^{\left(\frac{N-2}{2}\right)}(u_{ij})$, where $u_{ij}$ is a component of an orthogonal matrix, are eigenfunctions of $\Delta_{SO(N)}$ corresponding to the eigenvalues  $-\frac{k(k+N-2)}{2}$.
\medskip

In Section \ref{flag}, we introduce a flag of real vector spaces formed with products of power sum polynomials defined on $SO(N)$, flag that we prove is invariant under the Laplace-Beltrami operator $\Delta_{SO(N)}$. More precisely, the flag is  
$$V_0\subseteq V_{\leq 1}\subseteq \dots \subseteq V_{\leq k}\subseteq \dots,$$
with
$$V_{\leq k}:= \text{Span}_{\R}\{p_{_{\boldsymbol{\lambda}}}:SO(N)\to \R\,|\,\boldsymbol{\lambda}\vdash j\,\text{and}\,j\leq k\},$$
be the vector space generated by the products of power sum polynomials 
$$p_{_{\boldsymbol{\lambda}}}(U)=\tr(U^{m_1})\dots \tr(U^{m_s}),$$
where $\boldsymbol{\lambda}\vdash j$ is an integer partition of $j\in \mathbb{N}$.  In the literature, these polynomials are called trace polynomials and we will  refer to them as such. 

The main result in this section states that 
 $$\Delta_{SO(N)}(V_{\leq k})\subseteq V_{\leq k}.$$
 Choosing a basis for this flag, the matrix of $\Delta_{SO(N)}$ restricted to this flag takes an upper block triangular form,
\begin{equation*}
 \left[\Delta_{SO(N)}\right]_{\mathcal{V}}:=\left[\begin{array}{ccccc}
\Delta_{00} & \Delta_{01} & \dots & \Delta_{0k} & \dots \\
 \mathbb{O} & \Delta_{11} & \dots & \Delta_{1k} & \dots \\
\vdots & \vdots & \ddots & \vdots & \ddots \\
 \mathbb{O} &  \mathbb{O} & \dots & \Delta_{kk} & \dots \\
\vdots & \vdots & \ddots & \vdots & \ddots 
\end{array}\right],
\end{equation*}
where $\mathcal{V}:=\bigcup\limits_{k=0}^{\infty} V_{\leq k}$.

In Section \ref{cazul-SO(3)}, we apply the results obtained in the previous section to the case of the special orthogonal group $SO(3)$. More precisely, we compute two bases for the vector spaces $V_{\leq k}$, $k\in \mathbb{N}$, and write the corresponding matrices $\left[\Delta_{SO(3)}\right]_{V_{\leq k}}$.  This allows us to directly compute the eigenvalues of $\Delta_{SO(3)}$. Also, as a byproduct of having an explicit form for the matrix $\left[\Delta_{SO(3)}\right]_{V_{\leq k}}$, for any natural $k$, we give a formula for the irreducible characters of $SO(3)$ as trace polynomials. The irreducible characters of $SO(3)$ can be written using the above bases as follows:
\begin{align}\label{caractere-SO(3)-453}
\chi_k(U) & =\sum_{j=1}^k \tr(U^j)-k+1 \nonumber \\
 & =\sum_{j=0}^k \left(\sum_{l=j}^k(-1)^{k-l}C_{k+l}^{2l}C_{l}^j\right)(\tr(U))^j.
\end{align}

In the last section, we study the case of the special orthogonal group $SO(4)$. Since it is the only non-simple group in the family $SO(N)$, the study of the spectrum of $\Delta_{SO(4)}$ is a bit more delicate. For the case of simple, compact, connected Lie groups there exist elegant formulas that give all the eigenvalues of their respective Laplace-Beltrami operators, see \cite{svirkin}, \cite{svirkin-zubareva}, \cite{svirkin-2010},\cite{zubareva-2016}, \cite{zubareva-2020}, and \cite{cardona}.
We determine explicitly the spectrum of $\Delta_{SO(4)}$ by using the Riemannian submersion technique. More precisely, we prove
$$\text{spec}\Delta_{SO(4)}=\left\{-\frac{1}{4}\left(k_1(k_1+2)+k_2(k_2+2)\right)\,|\,k_1,k_2\in \mathbb{N}\,\,\hbox{and}\,\,k_1,k_2\,\,\text{with}\,\,\emph{\bf same parity}\right\}.$$
We use the notation from \cite{svirkin}, denoting by $\hbox{spec}\Delta_M$ the set of eigenvalues, without their multiplicities, of the Laplace-Beltrami operator defined on a Riemannian manifold $M$. 

Next, we construct a basis for $V_{\leq k}$, for all natural $k$. As a consequence, we give an explicit form of the matrix $\left[\Delta_{SO(4)}\right]_{V_{\leq 4}}$. We give a formula for the irreducible characters of $SO(4)$ up to order four, similar to \eqref{caractere-SO(3)-453}, as trace polynomials. Our approach can be applied to larger orders, with the computations being more involved.

\medskip
During the paper, we make certain notations, conventions and abuse of notations that we explain in the following.
For a matrix $U\in  \mathcal{M}_{N\times N}(\R)$, we denote by ${\bf u}_1,...,{\bf u}_N\in \R^N$ the  vectors formed with the columns of the matrix $U$ and consequently, $U$ has the form $U=\left[{\bf u}_1,...,{\bf u}_N\right]$. If $U\in SO(N)=\{U\in \mathcal{M}_{N\times N}(\R) \,|\,U^tU=UU^t=\mathbb{I}_N,\,\det U=1\}$, then the vectors ${\bf u}_1,...,{\bf u}_N\in \R^N$ are orthonormal.
We identify $\mathcal{M}_{N\times N}(\R)$ with $\R^{N^2}$ by the isomorphism $\text{vec}:\mathcal{M}_{N\times N}(\R)\rightarrow \R^{N^2}$ defined by the column vectorization $\text{vec}(U)\stackrel{\text{not}}{=}{\bf u}:=({\bf u}_1^t,...,{\bf u}_N^t)^t$.

In all that follows, in order to not overload the notations, we make a few {\bf conventions}:
\begin{itemize}
\item[(C$_1$)] We will sometimes denote with the same symbol a function defined on $\mathcal{M}_{N\times N}(\mathbb{R})$ with its restriction to the special orthogonal group $SO(N)$, and vice-versa, a function defined on the special orthogonal group $SO(N)$ with its (natural) prolongation defined on $\mathcal{M}_{N\times N}(\mathbb{R})$.
\item[(C$_2$)] When no confusions may arise, we will sometimes omit the subscript $Euc$ when writing $\nabla_{Euc}f, \text{Hess}_{Euc}f,\Delta_{Euc}f$, which denote the gradient, the Hessian, and respectively the Laplacian of a smooth function $f$ defined on $\mathcal{M}_{N\times N}(\mathbb{R})$, with respect to the Euclidean (Frobenius) metric ${\bf g}_{_{\text{Frob}}}$. 
When we refer to the gradient, Hessian, or Laplacian of a function defined on $(SO(N), {\bf g}_{_{\text{Frob}}})$, we will {\it always} use only the subscript $SO(N)$. 
\item[(C$_3$)] When we use the notations $\nabla_{Euc} f$ and $\text{Hess}_{Euc}f$ (with respect to the Euclidean metric on $\mathcal{M}_{N\times N}(\mathbb{R})$), we refer to them as being in {\it matrix} form, i.e., for a smooth function $f:\mathcal{M}_{N\times N}(\R)\to \R$, we denote $\widehat{f}:\R^{N^2}\to \R$, $\widehat{f}:=f\circ\text{vec}^{-1}$ and
$$\nabla_{Euc} f(U):=\text{vec}^{-1}\left(\nabla\widehat{f}(\text{vec}(U))\right),$$
$$\text{Hess}_{Euc}f(U):=[\text{Hess}\widehat{f}](\text{vec}(U))=\left[\begin{array}{ccc}
\left[\frac{\partial^2 \widehat{f}(\text{vec}(U))}{\partial {\bf u}_1\partial {\bf u}_1}\right] &\dots&\left[\frac{\partial^2 \widehat{f}(\text{vec}(U))}{\partial {\bf u}_1\partial {\bf u}_n}\right]\\
\vdots&\ddots&\vdots\\
\rule{0pt}{12pt}\left[\frac{\partial^2 \widehat{f}(\text{vec}(U))}{\partial {\bf u}_n\partial {\bf u}_1}\right]&\dots&\left[\frac{\partial^2 \widehat{f}(\text{vec}(U))}{\partial {\bf u}_n\partial {\bf u}_n}\right]\end{array}\right],$$
where
$$\left[\frac{\partial^2 \widehat{f}}{\partial {\bf u}_i\partial {\bf u}_j}\right]:=\left[\begin{array}{ccc}
\frac{\partial^2 \widehat{f}}{\partial {u}_{1i}\partial {u}_{1j}} &\dots&\frac{\partial^2 \widehat{f}}{\partial {u}_{1i}\partial {u}_{nj}}\\
\vdots&\ddots&\vdots\\
\rule{0pt}{12pt}\frac{\partial^2\widehat{f}}{\partial {u}_{ni}\partial {u}_{1j}}&\dots&\frac{\partial^2 \widehat{f}}{\partial {u}_{ni}\partial {u}_{nj}}\end{array}\right].$$
Also, $\Delta_{Euc} f(U):=\Delta\widehat{f}(\text{vec}(U))$.\\
This convention also applies when we write $\nabla_{SO(N)} f, \text{Hess}_{\,SO(N)}f, \Delta_{SO(N)}f$.
\end{itemize}

Using results from \cite{first-order} and \cite{second-order}, in \cite{laplacian} the following two formulas have been proved.

\begin{thm}\label{laplacian-sfera-10}\emph{(\cite{laplacian})}
For $\widetilde{f}:S_R^{N-1}\to \R$ a smooth function defined on the sphere of radius $R$ on $\R^N$, $f:\R^N\to \R$ a smooth prolongation of $\widetilde{f}$ and ${\bf x}\in S_R^{N-1}\subset \R^N$, we have the formula
\begin{equation}\label{Laplace-sfera-2022}
\left(\Delta_{S_R^{N-1}}\widetilde{f}\right)({\bf x})=\left(\Delta f\right)({\bf x})-\frac{1}{R^2}\tr\left({{\bf x}\,{\bf x}^t}\left[\emph{Hess}f\right]({\bf x})\right)-\frac{N-1}{R^2}\left<{\bf x}, \nabla f({\bf x})\right>.
\end{equation}
\end{thm}

\begin{thm}\label{teorema-principala}\emph{(\cite{laplacian})}
Let $\widetilde{f}:SO(N)\to \R$ be a smooth function and $f:\mathcal{M}_{N\times N}(\R)\to \R$ a smooth prolongation of $\widetilde{f}$. Then, for $U\in SO(N)$, we have 
\begin{equation}\label{laplacian-ort}
\Delta_{SO(N)}\widetilde{f}(U)=\frac{1}{2}\Delta f(U)-\frac{N-1}{2}\tr\left(U^t[\nabla f](U)\right)-\frac{1}{2}\tr\left(\Lambda(U)[\emph{Hess}f](U)\right),
\end{equation}
where the matrix $\Lambda(U)\in \mathcal{M}_{N^2\times N^2}(\R)$ is defined by 
\begin{equation*}\label{lambda-matrix}
\Lambda(U):=\left[\begin{array}{c|c|c}
{\bf u}_1 {\bf u}_1^t&\dots&{\bf u}_N {\bf u}_1^t\\
\hline
\vdots&\ddots&\vdots\\
\hline
\rule{0pt}{12pt}{\bf u}_1 {\bf u}_N^t&\dots&{\bf u}_N {\bf u}_N^t\end{array}\right].
\end{equation*}
\end{thm}

Applying the formula \eqref{Laplace-sfera-2022} to harmonic homogeneous polynomials, one obtains that $-\frac{k(k+N-2)}{R^2}$, $k\in \mathbb{N}$, are eigenvalues for $\Delta_{S_R^{N-1}}$. An important result proves that these are the only eigenvalues of $\Delta_{S_R^{N-1}}$. For details see \cite{axler}, among other classical references.

\section{Spherical eigenvalues of the Laplace-Beltrami operator on $SO(N)$}\label{spherical-eigenvalues-121}

In this section, we discuss the connection between the Laplace-Beltrami operator $\Delta_{SO(N)}$ and the Laplace-Beltrami operator on the sphere. As a consequence, we show that the values $-\frac{k(k+N-2)}{2}$, $k\in \mathbb{N}$, are among the eigenvalues of $\Delta_{SO(N)}$. Also, we prove that the Gegenbauer polynomials in any component $u_{ij}$ of an orthogonal matrix are eigenfunctions of $\Delta_{SO(N)}$ corresponding to the eigenvalues  $-\frac{k(k+N-2)}{2}$.

It is well known that a subset of eigenvalues for the Laplace-Beltrami operator on a Riemannian manifold $(M,{\bf g})$ can be obtained as eigenvalues of the Laplace-Beltrami operator on a (sometimes) simpler Riemannian manifold $(\overline{M}, \overline{{\bf g}})$ when there exists a Riemannian submersion $\pi:M\to \overline{M}$ such that all the fibers of $\pi$ are totally geodesic submanifolds of $M$, see \cite{wallach}, \cite{bergery}, \cite{svirkin}. In this case, the Laplace-Beltrami operators are linked by the following formula:
\begin{equation*}
(\Delta_{\overline{M}}h)\circ \pi=\Delta_M(h\circ\pi),\,\,\,h\in \mathcal{C}^{\infty}(\overline{M}).
\end{equation*}

As a straightforward consequence of this result, one has that all the eigenvalues of $\Delta_{\overline{M}}$ are also eigenvalues of $\Delta_{{M}}$. In general, this inclusion is strict. 

For the case of the special orthogonal group $(SO(N), {\bf g}_{_{\text{Frob}}})$ endowed with the induced Frobenius metric from $\mathcal{M}_{N\times N}(\R)$, we have the following Riemannian submersions depicted in the diagram.  

\medskip
\begin{center}
\begin{tikzcd}
                                                                                                &  &  & {\left(SO(N),{\bf g}_{_{\text{Frob}}}\right)} \arrow[d, "\pi"] \arrow[rrrd, "h\circ \pi"] \arrow[llld, "\widetilde{\pi}"'] &  &  &              \\
{\left( S^{N-1},2{\bf g}_{\hbox{can}}\right)} \arrow[rrr, "d_{\sqrt{2}}"'] &  &  & {\left(S^{N-1}_{\sqrt{2}},{\bf g}_{\hbox{can}}\right)} \arrow[rrr, "h"']                                     &  &  & {\mathbb{R}}
\end{tikzcd}
\end{center}
\medskip

The projection $\widetilde{\pi}$ takes a rotation $U=\left[ {\bf u}_1 \, \dots \, {\bf u}_N\right]\in SO(N)$ into its last column ${\bf u}_N\in S^{N-1}$. In order to make this projection a Riemannian submersion, on the sphere $S^{N-1}$, we have to consider the metric $2{\bf g}_{\text{can}}$ which is twice the induced Euclidean metric from $\R^N$. The dilation $d_{\sqrt{2}}:(S^{N-1},2{\bf g}_{\text{can}})\to (S^{N-1}_{\sqrt{2}},{\bf g}_{\text{can}})$ being an isometry makes the projection $\pi:(SO(N), {\bf g}_{_{\text{Frob}}})\to (S^{N-1}_{\sqrt{2}},{\bf g}_{\text{can}})$ a Riemannian submersion. The fibers of $\pi$ are isomorphic with $SO(N-1)$ and are totally geodesic submanifolds of $SO(N)$. 
As a consequence of the above considerations, one obtains the following well known result.

\begin{prop}\label{valori-proprii-sfera}
For all $k\in \mathbb{N}$, $-\frac{k(k+N-2)}{2}$ are among the eigenvalues of the Laplace-Beltrami operator $\Delta_{SO(N)}$.
\end{prop}

Next, we show how one can obtain the above result as a straightforward consequence of applying the formulas from Theorems  \ref{laplacian-sfera-10} and \ref{teorema-principala}.

Let $h:(S^{N-1}_{\sqrt{2}},{\bf g}_{\text{can}})\to \R$ a smooth function and $f:SO(N)\to \R$ defined by $f(U):=h(\sqrt{2}{\bf u}_N)$. 
In order not to burden the notations, we will use the same notations for the prolongations of $h$ and $f$ from their respective manifolds to the corresponding ambient spaces.
We obtain:
$$[\nabla f](U)=\sqrt{2}\left[\begin{array}{cccc} {\bf 0} & \dots & {\bf 0} & \nabla h(\sqrt{2}{\bf u}_N) \end{array}\right],$$ 
$$[\text{Hess} f](U)=2\left[\begin{array}{cccc} \mathbb{O}_N & \dots & \mathbb{O}_N &  \mathbb{O}_N \\
\vdots & \ddots & \vdots & \vdots \\ 
\mathbb{O}_N & \dots & \mathbb{O}_N &  \mathbb{O}_N \\
\mathbb{O}_N & \dots & \mathbb{O}_N & [\text{Hess}\, h](\sqrt{2}{\bf u}_N) \end{array}\right].$$ 
Also,
$$\tr(U^t [\nabla f](U))=\sqrt{2}\tr\left(\left[\begin{array}{c} {\bf u}^t_1 \\ \vdots \\  {\bf u}^t_N \end{array}\right] \left[\begin{array}{ccc} {\bf 0} & \dots & \nabla h(\sqrt{2}{\bf u}_N) \end{array}\right]\right)=\tr((\sqrt{2}{\bf u}_N)^t\nabla h(\sqrt{2}{\bf u}_N)).$$

\begin{align*}\tr(\Lambda(U)[\text{Hess}f](U)) & =2\tr\left(\left[\begin{array}{ccc} {\bf u}_1 {\bf u}_1^t & \dots & {\bf u}_N {\bf u}_1^t \\
\vdots & \ddots & \vdots \\ 
{\bf u}_1 {\bf u}_N^t & \dots & {\bf u}_N {\bf u}_N^t  \end{array}\right]\cdot \left[\begin{array}{cccc} \mathbb{O}_N & \dots & \mathbb{O}_N \\
\vdots & \ddots & \vdots \\ 
\mathbb{O}_N & \dots & [\text{Hess}\, h](\sqrt{2}{\bf u}_N) \end{array}\right]\right)\\
& = \tr ((\sqrt{2}{\bf u}_N)(\sqrt{2} {\bf u}_N)^t[\text{Hess}\, h](\sqrt{2}{\bf u}_N)).
\end{align*}

\begin{lemma}\label{lema-sfera-123}
Let $h:(S^{N-1}_{\sqrt{2}},{\bf g}_{\text{can}})\to \R$ a smooth function and $f:SO(N)\to \R$ defined by $f(U):=h(\sqrt{2}{\bf u}_N)$.
Then,
$$\Delta_{SO(N)}f(U)=\Delta_{S^{N-1}_{\sqrt{2}}}h(\sqrt{2}{\bf u}_N).$$
\end{lemma}

\begin{proof}
It is easy to observe that the Euclidean Laplacian of functions $f$ and $h$ are in the relation $\Delta f(U)=2\Delta h(\sqrt{2}{\bf u}_N)$. 

From Theorem \ref{teorema-principala} and the above computations we obtain,
\begin{align*}
\Delta_{SO(N)}f(U)= & \Delta h(\sqrt{2}{\bf u}_N)-\frac{N-1}{(\sqrt{2})^2} \tr\left((\sqrt{2}{\bf u}_N)^t\nabla h(\sqrt{2}{\bf u}_N)\right)\\
& -\frac{1}{(\sqrt{2})^2} \tr \left((\sqrt{2}{\bf u}_N) (\sqrt{2}{\bf u}_N)^t[\text{Hess}\, h](\sqrt{2}{\bf u}_N)\right).
\end{align*}
From the formula \eqref{Laplace-sfera-2022} the right side of the above equality is $\Delta_{S^{N-1}_{\sqrt{2}}}h(\sqrt{2}{\bf u}_N)$.
\end{proof}

This shows that the eigenvalues of the operator $\Delta_{S^{N-1}_{\sqrt{2}}}$, which are $-\frac{k(k+N-2)}{2}$, are also eigenvalues of the operator $\Delta_{SO(N)}$. Therefore, we have obtained another proof of Proposition \ref{valori-proprii-sfera}. 
The above Lemma holds true for any column ${\bf u}_j$ of the rotation matrix $U$.

Next, we will search for eigenfunctions of the particular form $f(U)= h(\sqrt{2}{\bf u}_j)=q(u_{ij})$, for $i,j\in\{1,\dots,N\}$ fixed and $q:[-1,1]\to \R$ a smooth function. From Lemma \ref{lema-sfera-123}, such a function $f$ can be an eigenfunction only for an eigenvalue of the form $-\frac{k(k+N-2)}{2}$, $k\in \mathbb{N}$. The equality $\Delta_{SO(N)}f=-\frac{k(k+N-2)}{2} f$, using formula  \eqref{laplacian-ort}, becomes 
\begin{equation*}
(1-u_{ij}^2)q''(u_{ij})-(N-1)u_{ij}q'(u_{ij})+k(k+N-2)q(u_{ij})=0,
\end{equation*} 
which is the differential equation of Gegenbauer type, see \cite{watson}. The only bounded solutions of this equation on $[-1,1]$ are the Gegenbauer polynomials 
$C_k^{\left(\frac{N-2}{2}\right)}(u_{ij})$, see \cite{otieno}. These polynomials are also intrinsically related to the representation theory of $SO(N)$, see the beautiful exposition in  \cite{vilenkin}.


\section{A flag of invariant subspaces for the Laplace-Beltrami operator $\Delta_{SO(N)}$}\label{flag}

We introduce a flag of trace polynomial functions that is let invariant by the Laplace-Beltrami operator. The Laplace-Beltrami operator restricted to this flag will take an upper block triangular form. 

We start this section by presenting a formula for the Laplace-Beltrami operator applied on a product of functions defined on a general Riemannian manifold.

\begin{prop}\label{general-product-99}
Let $(M,{\bf g})$ be a Riemannian manifold and  $f_1,\dots,f_k$ a set of smooth real functions on $M$. We have the following product rule: 
$$\Delta_{M}(f_1 \dots f_k)=\sum_{i=1}^k f_1\dots \widehat{f_i}\dots f_k \Delta_{M}f_i+2\sum_{1\leq i<j\leq k}f_1 \dots \widehat{f_i}\dots \widehat{f_j}\dots f_k\,{\bf g}(\nabla_{M}f_i,\nabla_{M}f_j),$$
where $~\widehat{\cdot}$ means that the factor is taken out from the product. 
\end{prop}

\begin{proof}
The proof is a direct induction starting from the standard formula valid on a general Riemannian manifold, $ \Delta_{M}(f_1 f_2)=f_1\Delta_{M}f_2+f_2\Delta_{M}f_1+2{\bf g}(\nabla_{M}f_1,\nabla_{M}f_2)$.
\end{proof}

For $m\in \mathbb{N}$, we introduce the {trace polynomial} $ p_m:SO(N)\to \R$ given by $p_m(U)=\tr(U^m)$. From this definition, $p_0$ is the constant function $p_0(U)=N$.

\begin{prop}\label{laplacian-tr-u-la-k}
The following relations hold:
\begin{enumerate}[(i)]
\item $\Delta_{SO(N)}p_{1}=-\dfrac{N-1}{2}p_1$ \emph{(\cite{laplacian})}.
\item 
$\Delta_{SO(N)}p_{m}=
\dfrac{m(1+(-1)^m)}{4} p_0+m\sum\limits_{i=0}^{\lfloor\frac{m-1}{2}\rfloor}p_{m-2i}-\frac{m(N+1)}{2}p_{m}-\frac{m}{2}\sum\limits_{j=1}^{m-1} p_{j}p_{m-j},$
 $m\geq 2$.
 \item For $q\geq 2$ a natural number, we have
$$\Delta_{SO(N)}(p_1^q)=-\dfrac{1}{2}\left((N-1)qp_1^q+q(q-1)(p_2-N)p_1^{q-2}\right).$$
\end{enumerate}
\end{prop}

\begin{proof} The proof of $(i)$ can be find in \cite{laplacian}.

$(ii)$ We apply the formula \eqref{laplacian-ort},  from \cite{laplacian}, to the case of $p_m$, using the same notation for the prolongation function, i.e. 
\begin{equation*}
\Delta_{SO(N)}p_m(U)=\frac{1}{2}\Delta p_m(U)-\frac{N-1}{2}\tr\left(U^t[\nabla p_m](U)\right)-\frac{1}{2}\tr\left(\Lambda(U)[\text{Hess}\,p_m](U)\right).
\end{equation*}

\noindent By Lemma \ref{trace-comm} and Lemma \ref{xun}, we have\\
$$\begin{aligned}
\dfrac{1}{2}\Delta p_m(U)&=\dfrac{1}{2}\tr([\hs p_m](U))
=\dfrac{1}{2}\tr\left(K_{NN}\left(m \sum_{r=0}^{m-2}\left(U^t\right)^r \otimes U^{m-r-2}\right)\right)\\
&=\dfrac{m}{2} \sum_{r=0}^{m-2}\tr(K_{NN}\left(U^t\right)^r \otimes U^{m-r-2})
=\dfrac{m}{2} \sum_{r=0}^{m-2}\tr(\left(U^t\right)^r\cdot U^{m-r-2}) \\
&= \begin{cases}
m( p_{m-2}+\dots+p_2)(U)+\frac{m}{2}p_0(U) & \text{if}\, m\, \text{is even} \\
m( p_{m-2}+\dots+p_1)(U) & \text{if}\, m\, \text{is odd}
\end{cases} \\
&= \begin{cases}
m(p_m+p_{m-2}+\dots+p_2)(U)-mp_m(U)+\frac{m}{2}p_0(U) & \text{if}\, m\, \text{is even} \\
m(p_m+p_{m-2}+\dots+p_1)(U)-mp_m(U) & \text{if}\, m\, \text{is odd}
\end{cases} \\
&= \dfrac{m(1+(-1)^m)}{4} p_0(U)+m\sum\limits_{i=0}^{\lfloor\frac{m-1}{2}\rfloor}p_{m-2i}(U)-mp_m(U).
\end{aligned}
$$
Also, we have, see \cite{cookbook},
$[\nabla p_m](U)=m(U^t)^{m-1},$
and therefore
$$-\dfrac{N-1}{2}\tr\left(U^t[\nabla p_m](U)\right)=-\dfrac{N-1}{2}\tr\left(m(U^t)^{m-1}\cdot U^t\right)=-\dfrac{m(N-1)}{2}p_m(U).$$
Finally, by Lemma \ref{xun} and Lemma \ref{lambda-comm}, 
$$
\begin{aligned}
    -\frac{1}{2}\tr\left(\Lambda(U) [\hs p_m](U)\right)
    &=-\frac{1}{2}\tr\left(\Lambda(U) \cdot K_{NN}\left(m \sum_{r=0}^{m-2}\left(U^t\right)^r \otimes U^{m-r-2}\right)\right)\\
    &=-\dfrac{m}{2}\sum_{r=0}^{m-2}\tr\left((U^t\otimes U)\left(\left(U^t\right)^r \otimes U^{m-r-2}\right)\right)\\
    &=-\dfrac{m}{2}\sum_{r=0}^{m-2}\tr\left((U^t)^{r+1}\otimes U^{m-r-1}\right)\\
    &=-\dfrac{m}{2}\sum_{r=0}^{m-2}p_{r+1}(U)p_{m-(r+1)}(U)
    =-\dfrac{m}{2}\sum_{j=1}^{m-1}p_j(U)p_{m-j}(U).
\end{aligned}
$$
Putting all together, we obtain the announced formula.

The proof of $(iii)$ is a consequence of Proposition \ref{general-product-99} and Lemma \ref{produs-grad}. 
\end{proof}

For $k\in \mathbb{N}^*$, an integer partition of $k$ is defined by a vector $ \boldsymbol{\lambda}:=(m_1, \dots, m_k)\in \mathbb{N}^k$ such that
$m_1\geq \dots \geq m_k\geq 0$ and $m_1+\dots+m_k=k$. Classically, it is denoted by $\boldsymbol{\lambda}\vdash k$. The partition function, denoted by $\mathcal{P}(k)$, represents the number of integer partitions of $k$. 

For every $\boldsymbol{\lambda}=(m_1, \dots, m_k)\vdash k$, we denote the trace polynomials
 $p_{_{\boldsymbol{\lambda}}}:SO(N)\to \R$ by
\begin{align*}
p_{_{\boldsymbol{\lambda}}}(U):=  \tr(U^{m_1})\dots \tr(U^{m_s}) 
=  p_{m_1}(U)\dots p_{m_s}(U),
\end{align*}
where $s:=\max\{i \,|\,m_i\geq 1\}$. When $\boldsymbol{\lambda}\neq (1,\dots,1)$, then we define  $r:=\max\{i \,|\,m_i\geq 2\}$.

Since $\tr(U^m)=\sum\limits_{i=1}^N\lambda_i^m$, where $\lambda_1,\dots, \lambda_N$ are the eigenvalues of $U$, we can write  $p_{_{\boldsymbol{\lambda}}}(U)$ as a product of power sums in $N$ complex variables.
{\bf The polynomials given by product of power sums in $N$ complex variables are algebraically independent over the field of complex numbers. However, when we restrict each of these variables $\lambda_i$ to the unit circle the independence of the polynomials $p_{\boldsymbol{\lambda}}$ is lost. 
}

 Starting with the real vector space
$$V_k:=\text{Span}_{\R}\{p_{_{\boldsymbol{\lambda}}}:SO(N)\to \R\,|\,\boldsymbol{\lambda}\vdash k\},$$
we introduce the following flag of vector spaces 
$$V_0\subseteq V_{\leq 1}\subseteq \dots \subseteq V_{\leq k}\subseteq \dots,$$
where 
$$V_{\leq k}:= \text{Span}_{\R}\{p_{_{\boldsymbol{\lambda}}}:SO(N)\to \R\,|\,\boldsymbol{\lambda}\vdash j\,\text{and}\,j\leq k\}.$$

\begin{thm}\label{flag-12}
The Laplace-Beltrami operator $\Delta_{SO(N)}$ leaves invariant the vector space $V_{\leq k}$, for all $k\in \mathbb{N}$; i.e.  $$\Delta_{SO(N)}(V_{\leq k})\subseteq V_{\leq k}.$$
\end{thm}

\begin{proof}
For $k=0$ the inclusion is trivial and for $k=1$, we have Proposition \ref{laplacian-tr-u-la-k}-$(i)$. Next, we distinguish several cases. 

We start with the case $r=s$, i.e. $p_{_{\boldsymbol{\lambda}}}(U):=\tr(U^{m_1})\dots \tr(U^{m_r})$, where $m_1\geq \dots\geq m_r\geq 2$.   Using Proposition \ref{general-product-99}, Proposition \ref{laplacian-tr-u-la-k}-$(ii)$ and Lemma \ref{scalar-gradienti-76}, we have
\begin{align*}
\Delta_{SO(N)} p_{_{\boldsymbol{\lambda}}} &=  \sum_{i=1}^r p_{m_1}\dots \widehat{p_{m_i}}\dots p_{m_r} \Delta_{SO(N)}p_{m_i} \\
& +2\sum_{1\leq i<j\leq r}p_{m_1} \dots \widehat{p_{m_i}}\dots \widehat{p_{m_j}}\dots p_{m_r}\,\left<\nabla_{SO(N)}p_{m_i},\nabla_{SO(N)}p_{m_j}\right> \\
& = \sum_{i=1}^r p_{m_1}\dots \widehat{p_{m_i}}\dots p_{m_r} \left(\dfrac{m_i(1+(-1)^{m_i})}{4} p_0+m_i\sum\limits_{l=0}^{\lfloor\frac{m_i-1}{2}\rfloor}p_{m_i-2l}\right.\\
& \left.-\frac{m_i(N+1)}{2}p_{m_i}-\frac{m_i}{2}\sum\limits_{j=1}^{m_i-1} p_{j}p_{m_i-j}\right) \\
& + \sum_{1\leq i<j\leq r}m_im_j p_{m_1} \dots \widehat{p_{m_i}}\dots \widehat{p_{m_j}}\dots p_{m_r}(p_{m_i-m_j}-p_{m_i+m_j}) \\
& = -\frac{k(N+1)}{2}p_{_{\boldsymbol{\lambda}}}+ N\sum_{i=1}^r \dfrac{m_i(1+(-1)^{m_i})}{4} p_{m_1}\dots \widehat{p_{m_i}}\dots p_{m_r} \\
& +
\sum_{i=1}^r m_i p_{m_1}\dots \widehat{p_{m_i}}\dots p_{m_r} \left(\sum\limits_{l=0}^{\lfloor\frac{m_i-1}{2}\rfloor}p_{m_i-2l}-
 \frac{1}{2}\sum\limits_{j=1}^{m_i-1} p_{j}p_{m_i-j}\right) \\
& + \sum_{1\leq i<j\leq r}m_im_j p_{m_1} \dots \widehat{p_{m_i}}\dots \widehat{p_{m_j}}\dots p_{m_r}(p_{m_i-m_j}-p_{m_i+m_j}).
\end{align*}
It is straightforward to see that every polynomial in the above sum belongs to $V_{\leq k}$.

\noindent For the case $\boldsymbol{\lambda}= (1,\dots,1)$, applying Proposition \ref{laplacian-tr-u-la-k}-$(iii)$, we obtain that $ \Delta_{SO(N)}(p_1^k)\in V_{\leq k}$.

In the mixed case, the polynomial  $p_{_{\boldsymbol{\lambda}}}(U):=\tr(U^{m_1})\dots \tr(U^{m_r})(\tr(U))^{s-r}$ can be written $ p_{_{\boldsymbol{\lambda}}}(U)=p_{_{\boldsymbol{\lambda}'}}(U)p_1^{s-r}$, where $p_{_{\boldsymbol{\lambda}'}}(U):=\tr(U^{m_1})\dots \tr(U^{m_r})\in V_{\leq k-s+r}$. Applying the formula for the Laplacian of the product, Proposition \ref{general-product-99}, and Lemma \ref{scalar-product-gradient-234}, we obtain
\begin{align*}
\Delta_{SO(N)} p_{_{\boldsymbol{\lambda}}}  = & (\Delta_{SO(N)} p_{_{\boldsymbol{\lambda}'}})p_1^{s-r}+ p_{_{\boldsymbol{\lambda}'}} \Delta_{SO(N)} p_1^{s-r}+2\left<\nabla_{SO(N)}p_{\boldsymbol{\lambda}'},\nabla_{SO(N)}p_{1}^{s-r}\right> \\
 = & (\Delta_{SO(N)} p_{_{\boldsymbol{\lambda}'}})p_1^{s-r}+ p_{_{\boldsymbol{\lambda}'}} \Delta_{SO(N)} p_1^{s-r} \\
& +(s-r)\sum_{i=1}^r m_ip_1^{s-r-1}p_{m_1}\dots\widehat{p_{m_i}}\dots p_{m_r}(p_{m_i-1}-p_{m_i+1}).
\end{align*}
From the first case,  $\Delta_{SO(N)} p_{_{\boldsymbol{\lambda}'}}\in V_{\leq k-s+r}$ and consequently, $(\Delta_{SO(N)} p_{_{\boldsymbol{\lambda}'}})p_1^{s-r}\in V_{\leq k}$. From Proposition \ref{laplacian-tr-u-la-k}-$(iii)$,  $\Delta_{SO(N)} p_1^{s-r}\in V_{\leq s-r}$  and consequently, $p_{_{\boldsymbol{\lambda}'}} \Delta_{SO(N)} p_1^{s-r}\in  V_{\leq k}$. It is easy to see that the last sum in the above formula also belongs to $V_{\leq k}$.
\end{proof}

We exemplify the formulas from  the above proof for the cases $k\in \{0,1,2,3,4\}$.
\begin{align*}
& k=0: \Delta_{SO(N)} p_{0}=0 \\
& k=1: \Delta_{SO(N)} p_{1}= -\frac{N-1}{2}p_1 \\
& k=2: \\
& \begin{cases}\Delta_{SO(N)} p_{(2,0)}=-(N-1)p_{(2,0)}-p_{(1,1)}+p_0 \\
\Delta_{SO(N)} p_{(1,1)}=-p_{(2,0)}-(N-1)p_{(1,1)}+p_0 \end{cases}\\
& k=3: \\
& \begin{cases}\Delta_{SO(N)} p_{(3,0,0)}=-\frac{3(N-1)}{2} p_{(3,0,0)}-3 p_{(2,1,0)}+3 p_1 \\
 \Delta_{SO(N)} p_{(2,1,0)}=-2 p_{(3,0,0)}-\frac{3(N-1)}{2} p_{(2,1,0)}-p_{(1,1,1)}+(N+2) p_1 \\ 
 \Delta_{SO(N)} p_{(1,1,1)}=-3 p_{(2,1,0)}-\frac{3(N-1)}{2} p_{(1,1,1)}+3 N p_{1}\end{cases} \\
& k=4: \\
& {\begin{cases}\Delta_{SO(N)} p_{(4,0,0,0)}=-2(N-1)p_{(4,0,0,0)}-4p_{(3,1,0,0)}-2p_{(2,2,0,0)}+4p_{(2,0)}+2p_0 \\
\Delta_{SO(N)} p_{(3,1,0,0)}=-3p_{(4,0,0,0)}-2(N-1)p_{(3,1,0,0)}-3p_{(2,1,1,0)}+3p_{(2,0)}+3p_{(1,1)} \\
 \Delta_{SO(N)} p_{(2,2,0,0)}=-4p_{(4,0,0,0)}-2(N-1)p_{(2,2,0,0)}-2p_{(2,1,1,0)}+2Np_{(2,0)}+4p_0 \\
\Delta_{SO(N)} p_{(2,1,1,0)}=-4p_{(3,1,0,0)}-p_{(2,2,0,0)} -2(N-1)p_{(2,1,1,0)}-p_{(1,1,1,1)}+Np_{(2,0)} 
+(N+4)p_{(1,1)}\\
 \Delta_{SO(N)} p_{(1,1,1,1)}=-6p_{(2,1,1,0)}-2(N-1)p_{(1,1,1,1)}+6Np_{(1,1)}.\end{cases}}
\end{align*}

The above theorem also shows that by choosing a basis $\mathcal{B}_{\leq k}$ for the vector space $V_{\leq k}$, the Laplace-Beltrami operator takes an upper block triangular form on the vector space 
\begin{equation}\label{space-V}
\mathcal{V}:=\bigcup\limits_{k=0}^{\infty} V_{\leq k}.
\end{equation}
 We denote 
\begin{equation}\label{forma-upper-block}
\left[\Delta_{SO(N)}\right]_{V_{\leq k}}:=\left[\begin{array}{cccc}
\Delta_{00} & \Delta_{01} & \dots & \Delta_{0k} \\
 \mathbb{O} & \Delta_{11} & \dots & \Delta_{1k} \\
\vdots & \vdots & \ddots & \vdots  \\
 \mathbb{O} & \mathbb{O} & \dots & \Delta_{kk}
\end{array}\right],
\left[\Delta_{SO(N)}\right]_{\mathcal{V}}:=\left[\begin{array}{ccccc}
\Delta_{00} & \Delta_{01} & \dots & \Delta_{0k} & \dots \\
 \mathbb{O} & \Delta_{11} & \dots & \Delta_{1k} & \dots \\
\vdots & \vdots & \ddots & \vdots & \ddots \\
 \mathbb{O} &  \mathbb{O} & \dots & \Delta_{kk} & \dots \\
\vdots & \vdots & \ddots & \vdots & \ddots 
\end{array}\right].
\end{equation}

From the general theory of upper block triangular matrices, see \cite{bernstein}, the set of eigenvalues of the matrix  $\left[\Delta_{SO(N)}\right]_{V_{\leq k}}$ is
$\bigcup_{m=0}^k {\cal E}^{(m)},$
where  ${\cal E}^{(m)}$ is the set of eigenvalues of the diagonal block $\Delta_{mm}$.\\


\section{The case of the orthogonal group $SO(3)$}\label{SO(3)}\label{cazul-SO(3)}

It is well known that the eigenvalues of Laplace-Beltrami operator for compact, connected, simply connected, and simple Lie groups are strongly related to the irreducible representations of these groups. Beautiful explicit formulas for the eigenvalues using representation theory are given in \cite{svirkin}, \cite{svirkin-zubareva}, \cite{svirkin-2010},\cite{zubareva-2016}, \cite{zubareva-2020}, and \cite{cardona}.

In this section, we compute two bases for the vector spaces $V_{\leq k}$, $k\in \mathbb{N}$, we write the corresponding matrices $\left[\Delta_{SO(3)}\right]_{V_{\leq k}}$, and determine the eigenvalues of $\Delta_{SO(3)}$. As a byproduct of having an explicit form for the matrix $\left[\Delta_{SO(3)}\right]_{V_{\leq k}}$, for any natural $k$, we give a formula for the irreducible characters of $SO(3)$ as trace polynomials.

Reminding that $p_0(U)=3$ and $p_1(U)=\tr(U)$, we start by proving the following result.

\begin{prop}
For the vector space  $V_{\leq k}$ we have:
\begin{enumerate}[(i)]
\item the set $\mathcal{B}'_{\leq k}:=\{p_0,p_1, \dots,p_1^k\}$ is a basis;
\item the set $\mathcal{B}''_{\leq k}:=\{p_0, p_1, \dots, p_k\}$ is a basis.
\end{enumerate}
\end{prop}

\begin{proof}
We have defined the vector space $V_{\leq k}$ as generated by the polynomials $p_{_{\boldsymbol{\lambda}}}$, more precisely
$V_{\leq k}:= \text{Span}_{\R}\{p_{_{\boldsymbol{\lambda}}}:SO(3)\to \R\,|\,\boldsymbol{\lambda}\vdash j\,\text{and}\,j\leq k\}.$\\
For a matrix $U\in SO(3)$, its eigenvalues are $1$, $\cos\alpha\pm i \sin\alpha$, with $\alpha\in [0,2\pi)$. It follows immediately that $p_1(U)=\tr(U)=2\cos\alpha+1$, leading to $\cos\alpha =\frac{p_1-1}{2}$.\\
For $m\geq 0$, we obtain $$p_m(U)=\tr(U^m)=1+2\cos(m \alpha)=1+2T_m(\cos \alpha)=1+2T_m\left(\frac{p_1-1}{2}\right),$$
where $T_m$ is the $m^{\hbox{{\footnotesize th}}}$ Chebyshev polynomial of first kind.


$(i)$  It follows that $p_m=f_m(p_1)$, with $f_m$ being a polynomial of degree $m$ in the variable $p_1$. If $f_m$ has a constant term $c$, then this term may be written as $\frac{c}{3}p_0$. Consequently, $\mathcal{B}_{\leq k}'$ is a set of generators for the vector space $V_{\leq k}$. 

Next, we prove that the set $\mathcal{B}_{\leq k}'$ is linearly independent. Suppose that there exist $\beta_0,\beta_1,\dots, \beta_k\in \R$ such that $\beta_0p_0+\beta_1p_1+\dots+\beta_kp_1^k=0$. Taking into account the structure of eigenvalues of an orthogonal $3\times 3$ matrix, this is equivalent with $3\beta_0+\beta_1(2x+1)+\dots+\beta_k(2x+1)^k=0$, for all $x=\cos\alpha\in [-1,1]$.
 Making the notation $Y=2x+1$, this is equivalent with the polynomial $3\beta_0+\beta_1 Y+\dots+\beta_k Y^k$ having an infinite number of roots. Therefore, all the coefficients $\beta_0, \dots, \beta_k$ are equal with zero. 
 
 $(ii)$ From $(i)$, we have that the dimension of $V_{\leq k}$ is $k+1$. Consequently, it is enough to prove the linear independence of the set $\mathcal{B}_{\leq k}''$. Suppose that there exist $\gamma_0,\gamma_1,\dots, \gamma_k\in \R$ such that $\gamma_0p_0+\gamma_1p_1+\dots+\gamma_kp_k=0$. This equality is equivalent with $(3\gamma_0+\gamma_1+\dots+\gamma_k)+2\gamma_1\cos(\alpha)+\dots+2\gamma_k\cos(k\alpha)=0$ for all $\alpha\in \R$ and thus $\gamma_0=\gamma_1=\dots=\gamma_k=0$.
\end{proof}

A consequence of the above proposition is that $\dim V_{\leq k}=\dim V_{\leq (k-1)}+1$ and thus the matrix $\left[\Delta_{SO(3)}\right]_{V_{\leq k}}$ written in the bases found above becomes an upper triangular matrix (every block $\Delta_{ij}$ is a scalar). 

Using the obvious relation $p_2=p_1^2-2p_1$  and the formula from Proposition \ref{laplacian-tr-u-la-k}-$(iii)$, we obtain 
$$\Delta_{SO(3)}(p_1^j)=-\dfrac{j(j+1)}{2}p_1^j+j(j-1)p_1^{j-1}+\frac{3}{2}j(j-1)p_1^{j-2},$$
for each $j$ with $2\leq j\leq k$.\\
Equivalently, the matrix of the operator $\Delta_{SO(3)}$ restricted to $V_{\leq k}$ written in the basis $\mathcal{B}_{\leq k}'$ is:
\footnotesize
\begin{equation*}
\left[ \Delta_{SO(3)}\right]_{V_{\leq k}}^{\mathcal{B}_{\leq k}'}:=\left[\begin{array}{ccccccccc}
0 & 0 & 3 & 0 &  \dots & 0 & 0 & 0 & 0 \\
0 & -1 & 2 & 9 &  \dots & 0 & 0 & 0 & 0 \\
0 & 0 & -3 & 6 & \dots & 0 & 0 & 0 & 0 \\
0 & 0 & 0 & -6 & \dots & 0 & 0 & 0 & 0 \\
\vdots & \vdots & \vdots & \vdots &  \ddots & \vdots & \vdots & \vdots & \vdots \\
0 & 0 & 0 & 0 & \dots & -\frac{(k-3)(k-2)}{2} & (k-3)(k-2) & \frac{3(k-2)(k-1)}{2} & 0 \\
0 & 0 & 0 & 0 & \dots & 0 & -\frac{(k-2)(k-1)}{2} & (k-2)(k-1) & \frac{3k(k-1)}{2} \\
0 & 0 & 0 & 0 & \dots & 0 & 0 & -\frac{k(k-1)}{2} &  k(k-1)\\
0 & 0 & 0 & 0 & \dots & 0 & 0 & 0 & -\frac{k(k+1)}{2}
\end{array}\right].
\end{equation*}\normalsize

For the second basis, in order to write $ \left[ \Delta_{SO(3)}\right]_{V_{\leq k}}^{\mathcal{B}_{\leq k}''}$, we need to further elaborate the formula from Proposition \ref{laplacian-tr-u-la-k} -$(ii)$ for the particular case of $SO(3)$.

\begin{lemma}
For the case of $SO(3)$, we have the following formula:
$$\Delta_{SO(3)}p_m=\frac{m(m-1)}{2}p_0-m\sum_{j=1}^{m-1}p_j-\frac{m(m+1)}{2}p_m,\,\,\,m\geq 2.$$
\end{lemma}

\begin{proof}
In the case of $SO(3)$, for $r,s\in \mathbb{N}$, we have
\begin{align*}
p_rp_s & = (1+2\cos(r\alpha))(1+2\cos(s\alpha)) \\
& = 1+2\cos(r\alpha)+2\cos(s\alpha)+2\cos((r+s)\alpha)+2\cos(|r-s|\alpha) \\
& = 1+(p_r-1)+(p_s-1)+(p_{r+s}-1)+(p_{|r-s|}-1) \\
& = p_r+p_s+p_{r+s}+p_{|r-s|}-p_0.
\end{align*}
Applying Proposition \ref{laplacian-tr-u-la-k}-$(ii)$ and the above computation, for the case $m$ odd,  we obtain:
\begin{align*}
\Delta_{SO(3)}p_m  = & m(p_{m-2}+p_{m-4}+\dots+p_1)-mp_m-\frac{m}{2}\sum_{j=1}^{m-1}p_jp_{m-j} \\
 = & m\left(p_{m-2}+p_{m-4}+\dots+p_1-p_m-\frac{1}{2}\sum_{j=1}^{m-1}\left(p_j+p_{m-j}+p_m+p_{|m-2j|}-p_0\right)\right) \\
 = & m\left((p_{m-2}+p_{m-4}+\dots+p_1)-\frac{m+1}{2}p_m-\sum_{j=1}^{m-1}p_j\right. \\
&  \left.+\frac{m-1}{2}p_0-(p_{m-2}+p_{m-4}+\dots+p_1)\right) \\
= & \frac{m(m-1)}{2}p_0-m\sum_{j=1}^{m-1}p_j-\frac{(m+1)m}{2}p_m.
\end{align*}
For the case $m$ even, by an analogous computation we obtain the announced result.
\end{proof}
The matrix of the operator $\Delta_{SO(3)}$ restricted to $V_{\leq k}$ written in the basis $\mathcal{B}_{\leq k}''$ is:
\begin{equation*}
{
\left[ \Delta_{SO(3)}\right]_{V_{\leq k}}^{\mathcal{B}_{\leq k}''}:=}\left[\begin{array}{ccccccccc}
0 & 0 & 1 & 3 &  \dots &\frac{(k-4)(k-3)}{2}  & \frac{(k-3)(k-2)}{2} & \frac{(k-2)(k-1)}{2} & \frac{k(k-1)}{2} \\
0 & -1 & -2 & -3 &  \dots & -(k-3) & -(k-2) & -(k-1) & -k \\
0 & 0 & -3 & -3 & \dots & -(k-3) & -(k-2) & -(k-1) & -k \\
0 & 0 & 0 & -6 & \dots & -(k-3) & -(k-2) & -(k-1) & -k \\
\vdots & \vdots & \vdots & \vdots &  \ddots & \vdots & \vdots & \vdots & \vdots \\
0 & 0 & 0 & 0 & \dots & -\frac{(k-3)(k-2)}{2} & -(k-2) & -(k-1) & -k \\
0 & 0 & 0 & 0 & \dots & 0 & -\frac{(k-2)(k-1)}{2} & -(k-1) & -k \\
0 & 0 & 0 & 0 & \dots & 0 & 0 & -\frac{k(k-1)}{2} & -k \\
0 & 0 & 0 & 0 & \dots & 0 & 0 & 0 & -\frac{k(k+1)}{2}
\end{array}\right].
\end{equation*}

As a consequence, analyzing both the matrices $\left[ \Delta_{SO(3)}\right]_{V_{\leq k}}^{\mathcal{B}_{\leq k}'}$ and $\left[ \Delta_{SO(3)}\right]_{V_{\leq k}}^{\mathcal{B}_{\leq k}''}$, we obtain that the numbers $\lambda_k=-\frac{k(k+1)}{2}$ are eigenvalues of the Laplace-Beltrami operator $\Delta_{SO(3)}$, for all $k\in \mathbb{N}$. These numbers constitute the entire spectrum of $\Delta_{SO(3)}$, as proved in \cite{svirkin-2010}; more precisely, for $\nu =2k+1$ and $\gamma=1$ the formula (10) in  \cite{svirkin-2010} gives $\lambda_k$. 

It has been proved in \cite{svirkin} that the characters of the  irreducible representations of a connected compact Lie group are eigenfunctions for the Laplace-Beltrami operator defined with respect to a biinvariant metric. In the following, we give a different proof for the case of $SO(3)$, where also we express the characters in the bases $\mathcal{B}_{\leq k}'$ and $\mathcal{B}_{\leq k}''$.

\begin{thm}
For any natural $k$, we have the following:
\begin{enumerate}[(i)]
\item The function $\chi_k:SO(3)\to \R$ defined by $$\chi_k:=-\frac{k-1}{3}p_0+p_1+p_2+\dots+p_k$$ is an eigenfunction for $\Delta_{SO(3)}$ corresponding to the eigenvalue $-\frac{k(k+1)}{2}$. 
\item The function $\chi_k$ has the equivalent expression 
$$\chi_k=\sum_{j=0}^k \left(\sum_{l=j}^k(-1)^{k-l}C_{k+l}^{2l}C_{l}^j\right)p_1^j.$$
\item The functions $\chi_k$ are the irreducible characters for the $SO(3)$ representations. 
\end{enumerate}
\end{thm}

\begin{proof} The statement $(i)$ is equivalent with proving that $v_k:=\left[\begin{array}{ccccc} -\frac{k-1}{3}&1&1&\dots&1\end{array}\right]^t$ is an eigenvector for $\left[ \Delta_{SO(3)}\right]_{V_{\leq k}}^{\mathcal{B}_{\leq k}''}$ corresponding to the eigenvalue $-\frac{k(k+1)}{2}$.

Indeed, the product of the first line of the matrix $\left[ \Delta_{SO(3)}\right]_{V_{\leq k}}^{\mathcal{B}_{\leq k}''}$ with the vector $v_k$ gives
$$\sum_{j=0}^{k-2}\frac{(k-j)(k-j-1)}{2}=\frac{k(k-1)(k+1)}{6}=-\frac{k(k+1)}{2}\cdot \left(-\frac{k-1}{3}\right)$$
and for  all $k\geq j\geq 1$ the product of the $(j+1)^{\text{th}}$ line of the matrix $\left[ \Delta_{SO(3)}\right]_{V_{\leq k}}^{\mathcal{B}_{\leq k}''}$ with the vector $v_k$ gives 
$$-\frac{j(j-1)}{2}-\sum_{l=0}^{k-j}(k-l)=-\frac{k(k+1)}{2}.$$

For $(iii)$, we notice that
$$\chi_k=1+2\sum_{m=1}^k\cos(m\alpha)=\frac{\sin\left(\left(k+\frac{1}{2}\right)\alpha\right)}{\sin\frac{\alpha}{2}},$$
which are the irreducible characters for the representations of $SO(3)$, see \cite{vilenkin}.

The statement $(ii)$ represents the expression of the eigenvector $\chi_k$ written in the basis $\mathcal{B}_k'$. In order to do this, we first notice that
$$\chi_k=U_{2k}\left(\cos \frac{\alpha}{2}\right),$$
where $U_{2k}$ are the even-indexed Chebyshev polynomials of the second kind, which are expressed only in even powers of the variable. More precisely, it is known that
$$U_{2k}(t)=\sum_{s=0}^k(-1)^sC_{2k-s}^s~4^{k-s}(t^2)^{k-s}.$$
Noticing that $\cos^2\frac{\alpha}{2}=\frac{1+\cos \alpha}{2}=\frac{p_1+1}{4}$ and by a change of the summation index, we obtain
$$\chi_k=\sum_{l=0}^k(-1)^{k-l}C_{k+l}^{2l}\cdot(p_1+1)^l.$$
Taking into account the coefficient of $p_1^j$, $0\leq j\leq k$, in the binomial expansion of $(p_1+1)^l$, we obtain the result from $(ii)$.
\end{proof}

\section{The case of the orthogonal group $SO(4)$}\label{SO(4)-123}

Since $SO(4)$ is the only non-simple group in the family $SO(N)$, the study of $\text{spec}\Delta_{SO(4)}$ is a bit more delicate. 
In the first subsection,  we determine explicitly the spectrum of $\Delta_{SO(4)}$ by using the Riemannian submersion technique. 
In the second subsection, we construct a basis for $V_{\leq k}$, for all natural $k$. As a consequence, we give an explicit form of the matrix $\left[\Delta_{SO(4)}\right]_{V_{\leq 4}}$. We also give a formula for the irreducible characters of $SO(4)$ up to order four, similar with the formula \eqref{caractere-SO(3)-453} for the case $SO(3)$.

\subsection{Eigenvalues of Laplace-Beltrami operator on $SO(4)$}

In order to establish the set of eigenvalues of Laplace-Beltrami operator on $SO(4)$, we need some preliminary results. The following map 
$\Phi:S^3\to SU(2)$, 
$$\Phi(a,b,c,d)=\begin{pmatrix} a+ i b & c+i d \\ -c +i d  & a- i b \end{pmatrix}$$
is a Lie group isomorphism. We search for a Riemannian metric on $S^3$ of the form $\alpha {\bf g}_{_{\text{Can}}}$, $\alpha>0$, so that $\Phi $ becomes an isometry, where on $SU(2)$ we consider the metric ${\bf g}_{_{\text{Frob}}}$. For this it is sufficient to study the differential $d\Phi(1,0,0,0): T_{(1,0,0,0)}S^3\to su(2)$, 
$$d\Phi(1,0,0,0)\cdot (0,v_2,v_3,v_4)=\begin{pmatrix} i v_2 & v_3+i v_4 \\ -v_3 +i v_4  & - i v_2 \end{pmatrix}.$$
\begin{align*}
{\bf g}_{_{\text{Frob}}}&(d\Phi(1,0,0,0)\cdot (0,v_2,v_3,v_4), d\Phi(1,0,0,0)\cdot (0,w_2,w_3,w_4)) \\
& = 
\tr((d\Phi(1,0,0,0)\cdot (0,v_2,v_3,v_4))^{\dagger}d\Phi(1,0,0,0)\cdot (0,w_2,w_3,w_4) \\
& =
2(v_2w_2+v_3w_3+v_4w_4) \\
& =2{\bf g}_{_{\text{Can}}}((0,v_2,v_3,v_4),(0,w_2,w_3,w_4)).
\end{align*}
The above equality proves that $\Phi:(S^3,2{\bf g}_{_{\text{Can}}})\to (SU(2),{\bf g}_{_{\text{Frob}}})$ is an isometry.
From \cite{svirkin}, we have the operator equality 
$$\Delta_{(M,k{\bf g})}=\frac{1}{k}\Delta_{(M,{\bf g})},$$  
where $(M,{\bf g})$ is a Riemannian manifold and $k>0$ is a positive scalar.
Consequently, 
$$(\Delta_{SU(2)}f)\circ \Phi=\Delta_{(S^3,2{\bf g}_{_{\text{Can}}})}(f\circ \Phi)=\frac{1}{2}\Delta_{S^3}(f\circ \Phi), 
\,\,f\in \mathcal{C}^{\infty}(SU(2)).$$
As it is shown in \cite{axler}, the eigenvalues of $\Delta_{S^3}$ are $\{-k(k+2)\,|\,k\in \mathbb{N}\}$ and thus, 
$$\text{spec}\Delta_{SU(2)}=\left\{-\frac{1}{2}k(k+2)\,|\,k\in \mathbb{N}\right\}.$$ 

It is well known that the eigenvalues of the Laplace-Beltrami operator on a compact, connected Lie group are strongly related to its irreducible representations, see for example \cite{vilenkin} or other classical texts.

In \cite{svirkin}, pp. 240, it is proved that in order to compute $\text{spec}\Delta_{G}$, where $G$ is a compact, connected Lie group, it is sufficient to apply the Laplace-Beltrami operator on $G$ to every irreducible character of $G$. In other words, the set of eigenvalues (without multiplicity) of $\Delta_{G}$ is in one to one correspondence with the set of irreducible characters of $G$. In the light of this important result, we will denote with $\chi_{\frac{k}{2}}$, $k\in \mathbb{N}$, the irreducible character of $SU(2)$ corresponding to the eigenvalue $-\frac{1}{2}k(k+2)=-2\left[\frac{k}{2}\left(\frac{k}{2}+1\right)\right]$. Making the notation $j:=\frac{k}{2}$, we can rewrite the irreducible character as $\chi_j$ with the corresponding eigenvalue $-2j(j+1)$, $j\in \frac{1}{2}\mathbb{N}$, which is the standard notation in representation theory.

There is a tight relation between $SU(2)\times SU(2)$ and $SO(4)$. More explicitly, see \cite{fujii}, \cite{makhlin},  let $F:SU(2)\times SU(2)\to SO(4)$
$$F(X,Y)=R^{\dagger}(X\otimes Y)R,\,\,\text{where}\,\,
R=\frac{1}{\sqrt{2}}\begin{pmatrix} 1 & 0 & 0 & -i \\ 0 & -i & -1 & 0 \\ 0 & -i & 1 & 0 \\ 1 & 0 & 0 & i
\end{pmatrix},$$
is a Lie group morphism. 
This time, we search for a Riemannian metric on $SU(2)\times SU(2)$ of the form $a ({\bf g}_{_{\text{Frob}}}\times  {\bf g}_{_{\text{Frob}}})$, $a>0$, so that $F $ becomes a Riemannian submersion, where on $SO(4)$ we consider the metric ${\bf g}_{_{\text{Frob}}}$. As before, it is sufficient to study the differential $dF(\mathbb{I}_2,\mathbb{I}_2): su(2)\times su(2)\to so(4)$, which is given by 
$$dF(\mathbb{I}_2,\mathbb{I}_2)\cdot (\boldsymbol{\xi}_1, \boldsymbol{\xi}_2)=R^{\dagger}(\boldsymbol{\xi}_1\otimes \mathbb{I}_2+\mathbb{I}_2\otimes \boldsymbol{\xi}_2)R.$$
If $\boldsymbol{\xi}_j=\begin{pmatrix} i\alpha_j & -\overline{z}_j \\ z_j & -i \alpha_j\end{pmatrix}$ and $\boldsymbol{\eta}_j=\begin{pmatrix} i\beta_j & -\overline{w}_j \\ w_j & -i \beta_j\end{pmatrix}$ are matrices in $su(2)$, $\alpha_j,\beta_j\in \R$, $z_j, w_j\in \mathbb{C}$, and $j\in{1,2}$, then 
\begin{align*}
{\bf g}_{_{\text{Frob}}}&\left(dF(\mathbb{I}_2,\mathbb{I}_2)\cdot (\boldsymbol{\xi}_1, \boldsymbol{\xi}_2), dF(\mathbb{I}_2,\mathbb{I}_2)\cdot (\boldsymbol{\eta}_1, \boldsymbol{\eta}_2)\right) \\
& = 
\tr\left((dF(\mathbb{I}_2,\mathbb{I}_2)\cdot (\boldsymbol{\xi}_1, \boldsymbol{\xi}_2))^{\dagger}dF(\mathbb{I}_2,\mathbb{I}_2)\cdot (\boldsymbol{\eta}_1, \boldsymbol{\eta}_2)\right) \\
& =
4\alpha_1\beta_1+4\alpha_2\beta_2+2\overline{z}_1w_1+2z_1\overline{w}_1+2\overline{z}_2w_2+2z_2\overline{w}_2 \\
& =2(\tr(\boldsymbol{\xi}_1^{\dagger}\boldsymbol{\eta}_1)+\tr(\boldsymbol{\xi}_2^{\dagger}\boldsymbol{\eta}_2)) \\
& =2({\bf g}_{_{\text{Frob}}}\times {\bf g}_{_{\text{Frob}}})((\boldsymbol{\xi}_1, \boldsymbol{\xi}_2),(\boldsymbol{\eta}_1, \boldsymbol{\eta}_2)).
\end{align*}
The above equality proves that $F:(SU(2)\times SU(2), 2({\bf g}_{_{\text{Frob}}}\times {\bf g}_{_{\text{Frob}}})\to (SO(4), {\bf g}_{_{\text{Frob}}})$ is a Riemannian submersion. And so, see \cite{svirkin}, for every $f\in \mathcal{C}^{\infty}(SO(4))$, we have 
$$(\Delta_{SO(4)}f)\circ F=\Delta_{\left(SU(2)\times SU(2), 2({\bf g}_{_{\text{Frob}}}\times {\bf g}_{_{\text{Frob}}})\right)}(f\circ F)=\frac{1}{2}\Delta_{SU(2)\times SU(2)}(f\circ F).$$

The irreducible representations of $SU(2)\times SU(2)$ are given by the tensor product of the irreducible representations of $SU(2)$. The irreducible characters are the products $\chi_{j_1}\chi_{j_2}$, $j_1, j_2\in \frac{1}{2}\mathbb{N}$. Their corresponding eigenvalues are the result of the following computation. For any $(X,Y)\in SU(2)\times SU(2)$, we have 
\begin{align*}
 \Delta_{SU(2)\times SU(2)}(\chi_{j_1}(X)\chi_{j_2}(Y)) & =(\Delta_{_{SU(2)}}\chi_{j_1}(X))\chi_{j_2}(Y)+\chi_{j_1}(X)(\Delta_{_{SU(2)}}\chi_{j_2}(Y)) \\
& =- \left(2j_1(j_1+1)+2j_2(j_2+1)\right)\chi_{j_1}(X)\chi_{j_2}(Y).
\end{align*}

The irreducible representations of $SO(4)$ are well known and studied for a long time, see \cite{kosmann}. They are obtained from the irreducible representations of $SU(2)$.  More precisely, the irreducible representations of $SO(4)$ are $D^{(j_1j_2)}:=D^{j_1}\otimes D^{j_2}$ that can be factorized through the map $F$, where $D^{j}$, $j\in \frac{1}{2}\mathbb{N}$, are the irreducible representations of $SU(2)$. The equivalent condition for $D^{(j_1j_2)}$ to factorize through the map $F$ is 
\begin{equation}\label{conditie-j}
j_1, j_2\in \frac{1}{2}\mathbb{N}\,\,\text{and}\,\, j_1+j_2\in \mathbb{N}.
\end{equation}
The above condition is equivalent with the following 
\begin{equation}\label{conditie-k}
j_1=\frac{k_1}{2}, j_2=\frac{k_2}{2}, k_1, k_2\in \mathbb{N}\,\,\text{and}\,\, k_1,\,k_2\,\,\text{have the same parity}.
\end{equation}


\begin{thm}\label{spectru-SO4-789}
Let the special orthogonal group $(SO(4), {\bf g}_{_{\emph{Frob}}})$ endowed with the induced Frobenius metric from $\mathcal{M}_{4\times 4}(\R)$ and $\Delta_{SO(4)}$ be its associated Laplace-Beltrami operator.   
The set of eigenvalues of $\Delta_{SO(4)}$ is given by
$$\emph{spec}\Delta_{SO(4)}=\left\{-\frac{1}{4}(k_1(k_1+2)+k_2(k_2+2))\,|\,k_1,k_2\in \mathbb{N}\,\,\emph{and}\,\,k_1,k_2\,\,\emph{with}\,\,\emph{\bf same parity}\right\}.$$
\end{thm}

\begin{proof}
As proved in \cite{svirkin}, in order to compute $\text{spec}\Delta_{SO(4)}$ it is sufficient to apply the Laplace-Beltrami operator on $SO(4)$ to every irreducible character of $SO(4)$.

 The characters of irreducible representations of $SO(4)$ are $\chi_{j_1,j_2}$, where $(\chi_{j_1,j_2}\circ F)(X,Y)=\chi_{j_1}(X)\chi_{j_2}(Y)$ and $j_1,j_2$ satisfy \eqref{conditie-j}, or equivalently \eqref{conditie-k}.

According to \cite{svirkin}, to each character $\chi_{j_1,j_2}$ it corresponds a unique eigenvalue $\lambda_{j_1,j_2}$ of $\Delta_{SO(4)}$. From
$\lambda_{j_1,j_2}(\chi_{j_1,j_2})  =\Delta_{SO(4)}\chi_{j_1,j_2}$ composing both sides with $F$, we obtain
\begin{align*}
\lambda_{j_1,j_2}(\chi_{j_1,j_2}\circ F) & =(\Delta_{SO(4)}\chi_{j_1,j_2})\circ F=\frac{1}{2}\Delta_{SU(2)\times SU(2)}(\chi_{j_1,j_2}\circ F)\\
& = \frac{1}{2}\Delta_{SU(2)\times SU(2)}(\chi_{j_1}\chi_{j_2}) 
 =- \frac{1}{2}\left(2j_1(j_1+1)+2j_2(j_2+1)\right)\chi_{j_1}\chi_{j_2}\\
& = -\left(j_1(j_1+1)+j_2(j_2+1)\right) (\chi_{j_1, j_2}\circ F) \\
& = -\frac{1}{4} \left(k_1(k_1+2)+k_2(k_2+2)\right) (\chi_{j_1, j_2}\circ F),
\end{align*}
which proves the announced result.
\end{proof}

\subsection{Characters of $SO(4)$ as eigenvectors of Laplace-Beltrami operator} 

The purpose of this subsection is to write the characters of $SO(4)$ as trace polynomials. 
As in Section \ref{SO(3)}, we compute a basis for the vector spaces $V_{\leq k}$, $k\in \mathbb{N}$. As an example, we write the matrix $\left[\Delta_{SO(4)}\right]_{V_{\leq k}}$, up to order four, in this basis. This allows us to determine the eigenvalues and eigenvectors of $\left[\Delta_{SO(4)}\right]_{V_{\leq 4}}$.

Recalling that $p_0(U)=4$, $p_1(U)=\tr(U)$, and $p_2(U)=\tr(U^2)$, we give a basis for $V_{\leq k}$.

\begin{prop}
The following set of products of trace polynomials $$\mathcal{B}_{\leq k}:=\{p_0\}\cup \{p_1^lp_2^m\,|\,l,m\in \mathbb{N}\,\,\text{and}\,\,0<l+2m\leq k\}$$ is a basis for the vector space $V_{\leq k}$.
\end{prop}

\begin{proof}
Recall the definition
$V_{\leq k}:= \text{Span}_{\R}\{p_{_{\boldsymbol{\lambda}}}:SO(4)\to \R\,|\,\boldsymbol{\lambda}\vdash j\,\text{and}\,j\leq k\}.$\\
For a matrix $U\in SO(4)$, its eigenvalues are $\cos\alpha\pm i \sin\alpha$, $\cos\beta\pm i \sin\beta$, with $\alpha, \beta\in [0,2\pi)$.  Consequently,  
\begin{align*}
p_1(U) & =\tr(U)=2(\cos\alpha+\cos\beta),\\
p_2(U) & =\tr(U^2)=2(\cos 2\alpha+\cos 2\beta)=4(\cos^2\alpha+\cos^2\beta-1), \\
p_3(U) & =\tr(U^3)=2(\cos 3\alpha+\cos 3\beta)= 8 \cos^3\alpha+  8 \cos^3\beta-6(\cos\alpha+\cos\beta) \\
 & = 2(\cos\alpha+\cos\beta)((2\cos\alpha)^2-4\cos\alpha\cos\beta+(2\cos\beta)^2)-6(\cos\alpha+\cos\beta) \\
 & = p_1(U)\left(p_2(U)+4-\frac{1}{2}(p_1^2(U)-p_2(U)-4)\right)-3p_1(U) \\
 & = -\frac{1}{2}p_1^3(U)+\frac{3}{2}p_1(U)p_2(U)+3p_1(U).
\end{align*}
Observing that $p_3$ can be written as a function of $p_1$ and $p_2$, we proceed by induction showing that any $p_s$ can also be written as a function of $p_1$ and $p_2$ for any $s\in \mathbb{N}^*$.

From the Cayley-Hamilton equality for a matrix $U\in SO(4)$, we have
$$U^4-p_1U^3+\frac{1}{2}\left(p_1^2-p_2\right)U^2-\frac{1}{6}\left(p_1^3-3p_1p_2+2p_3\right)U+ \mathbb{I}_4=\mathbb{O}_4.$$
Using the formula of $p_3$ as a function of $p_1,p_2$, we get the equivalent equality
 $$U^4-p_1U^3+\frac{1}{2}\left(p_1^2-p_2\right)U^2-p_1U+ \mathbb{I}_4=\mathbb{O}_4.$$
Assuming that $p_1,p_2,\dots,p_s$ can be written as function of $p_1,p_2$, we prove that $p_{s+1}$ can also be written as function of $p_1,p_2$. Indeed, multiplying the Cayley-Hamilton equality by $U^{s-3}$ and taking the trace, we obtain
$$p_{s+1}=p_1p_s-\frac{1}{2}\left(p_1^2-p_2\right)p_{s-1}+p_1p_{s-2}-p_{s-3}.$$
This shows that $\mathcal{B}_{\leq k}$ is a set of generators for the vector space $V_{\leq k}$. 

In order to prove the linear independence of $\mathcal{B}_{\leq k}$, suppose $c_0p_0+\sum c_{lm}p_1^lp_2^m=0$. Making the notation $\cos\alpha=x$ and $\cos \beta=y$, we have equivalently 
$$4c_0+\sum c_{lm}(2(x+y))^l\cdot (4(x^2+y^2-1))^m=0,~\forall x\in [-1,1],\forall y\in [-1,1].$$
Assigning $x=y=0$, we obtain that $c_0=0$.
We can easily prove that for any $a\in [-1,1],~b\in [-1,0]$ there exist $x\in [-1,1]$, $y\in [-1,1]$ such that $2(x+y)=a,~4(x^2+y^2-1)=b$.
It follows that
$$\sum c_{lm }a^l b^m=0,~\forall a\in [-1,1],\forall b\in [-1,0].$$
Fixing $a\in [-1,1]$ arbitrarily, the polynomials in the variable $b$, 
$$\sum_m \left(\sum_l c_{lm}a^l\right)b^m$$
have each an infinite number of roots, therefore their coefficients $\sum\limits_l c_{lm}a^l$, which are polynomials in the variable $a$, have also an infinite number of roots. This leads to $c_{lm}=0$ for all $l,m$.
\end{proof}

In order to prove that the spectrum of $\Delta_{SO(4)}$ is the same as the spectrum of the restricted operator $\Delta_{SO(4)}|_{\mathcal{V}}:\mathcal{V}\to \mathcal{V}$, see Theorem \ref{flag-12} and notation \eqref{space-V}, we need the following technical result.

\begin{lemma}\label{characters-in-V}
The irreducible characters of $SO(4)$ representations are vectors in $\mathcal{V}$.
\end{lemma}

The idea for the proof is to show that every irreducible character of $SO(4)$  can be written as a linear combination of elements of the form $p_1^l p_2^m$. For the proof, see Appendix.

The set $\text{spec}\Delta_{SO(4)}$, obtained in Theorem \ref{spectru-SO4-789}, can be obtained inductively.
\begin{thm} 
We have the following equality:
$$\emph{spec}\Delta_{SO(4)}=\bigcup\limits_{k\geq 0} \emph{spec} \left[\Delta_{SO(4)}\right]_{V_{\leq k}}.$$
\end{thm}

\begin{proof}
The inclusion from right to left is obvious.
For the other inclusion, we start with $\lambda\in \text{spec}\Delta_{SO(4)}$. By Theorem 3.1. from \cite{svirkin}, there exists an irreducible character $\chi_{_{\lambda}}$ that is an eigenvector, i.e. $\Delta_{SO(4)}\chi_{_{\lambda}}=\lambda \chi_{_{\lambda}}$. From Lemma \ref{characters-in-V}, we have $\chi_{_{\lambda}}\in V_{\leq k}$ for some $k\in \mathbb{N}$. This concludes the proof. 
\end{proof}

In order to compute  the matrix $\left[\Delta_{SO(4)}\right]_{V_{\leq k}}$ in the basis $\mathcal{B}_{\leq k}$, see \eqref{forma-upper-block}, first we need to give a formula for $\Delta_{SO(4)}(p_1^lp_2^m) $ expressed as a linear combination of elements from the basis $\mathcal{B}_{\leq k}$. By particularizing the computations from the proof of Theorem \ref{flag-12}, we obtain
\begin{align*}
\Delta_{SO(4)}(p_1^lp_2^m)   = & m \left(m -1\right) p_{1}^{l} p_{2}^{m -2} \left(p_{1}^2-4\right)^{2}+m p_{1}^{l} p_{2}^{m -1} \left(\left(l -2 m +1\right) p_{1}^{2}-4 l +4\right) \\
& -\frac{l \left(l -1\right) p_{1}^{l -2} p_{2}^{m} \left(p_{2}-4\right)}{2}-\frac{\left(6 m l +2 m^{2}+3 l +4 m \right)  p_{1}^{l}p_{2}^{m}}{2}.
\end{align*}

We also need to choose a total order relation for the elements of $\mathcal{B}_{\leq k}$:
$p_1^{l'}p_2^{m'}\preceq p_1^l p_2^m$ if and only if ($l'+2m'<l+2m$) or ($l'+2m'=l+2m$ and $m'\leq m$).
For $k=4$, we have the following components of the matrix
$\left[ \Delta_{SO(4)}\right]_{V_{\leq 4}}^{\mathcal{B}_{\leq 4}}$,
$$\begin{array}{c|ccccccccc}
&\Delta p_0 & \Delta p_1 & \Delta p_1^2 & \Delta p_2 & \Delta p_1^3 & \Delta p_1p_2 & \Delta p_1^4 & \Delta p_1^2 p_2 & \Delta p_2^2 \\
\hline
p_0 & 0 & 0 & 1 & 1 & 0 & 0 & 0 & 0 & 8 
\\
p_1 & 0 & -\frac{3}{2} & 0 & 0 & 12 & 0 & 0 & 0 & 0 
\\
p_1^2 & 0 & 0 & -3 & -1 & 0 & 0 & 24 & -4 & -16 
\\
p_2 & 0 & 0 & -1 & -3 & 0 & 0 & 0 & 4 & 8 
\\
p_1^3 & 0 & 0 & 0 & 0 & -\frac{9}{2} & 0 & 0 & 0 & 0 
\\
p_1p_2 & 0 & 0 & 0 & 0 & -3 & -\frac{15}{2} & 0 & 0 & 0 
\\
 p_1^4 &0 & 0 & 0 & 0 & 0 & 0 & -6 & 1 & 2 
\\
p_1^2 p_2 & 0 & 0 & 0 & 0 & 0 & 0 & -6 & -12 & -6 
\\
 p_2^2 & 0 & 0 & 0 & 0 & 0 & 0 & 0 & -1 & -8 
\end{array}$$
The diagonal blocks are $\Delta_{00}=[0],\,\,\Delta_{11}=\left[\begin{array}{c}-\frac{3}{2}\end{array}\right]$,
$$\Delta_{22}=\left[\begin{array}{cc} -3 & -1 \\ -1 & -3\end{array}\right],\,\,
\Delta_{33}=\left[\begin{array}{cc} -\frac{9}{2} & 0 \\ -3 & -\frac{15}{2}\end{array}\right],\,\, \Delta_{44}=\left[\begin{array}{ccc} -6 & 1 & 2 \\ -6 & -12 & -6 \\ 0 & -1 & -8\end{array}\right].$$
The eigenvalues of $\Delta_{SO(4)}$ obtained from these blocks are 
$$\text{spec} \left[\Delta_{SO(4)}\right]_{V_{\leq 4}}=\left\{0, -\frac{3}{2}, -2, -4, -\frac{9}{2},-\frac{15}{2},-6, -8, -12\right\},$$
with the corresponding eigenvectors (characters)
\begin{align*}
& {\bf v}_0  =(1,0,0,0,0,0,0,0,0)=p_0\\
& {\bf v}_{-\frac{3}{2}}  =(0,2,0,0,0,0,0,0,0)=2p_1 \\
& {\bf v}_{-2}  =\left(0,0,\frac{1}{2},-\frac{1}{2},0,0,0,0,0\right)=\frac{1}{2}p_1^2-\frac{1}{2}p_2 \\
& {\bf v}_{-4}  =\left(-\frac{1}{2},0,1,1,0,0,0,0,0\right)=-\frac{1}{2}p_0+p_1^2+p_2 \\
& {\bf v}_{-\frac{9}{2}} =\left(0,-2,0,0,\frac{1}{2},-\frac{1}{2},0,0,0\right)=-2p_1+\frac{1}{2}p_1^3-\frac{1}{2}p_1p_2\\
& {\bf v}_{-\frac{15}{2}}  =(0,0,0,0,0,2,0,0,0)=2p_1p_2\\
& {\bf v}_{-6}  =\left(0,0,-\frac{3}{2},-\frac{1}{2},0,0,\frac{1}{4},-\frac{1}{2},\frac{1}{4}\right)=-\frac{3}{2}p_1^2-\frac{1}{2}p_2+\frac{1}{4}p_1^4-\frac{1}{2}p_1^2p_2+\frac{1}{4}p_2^2 \\ 
& {\bf v}_{-8} =\left(\frac{1}{2},0,-2,0,0,0,\frac{1}{4},0,-\frac{1}{4}\right)=\frac{1}{2}p_0-2p_1^2+\frac{1}{4}p_1^4-\frac{1}{4}p_2^2\\
& {\bf v}_{-12} =\left(-\frac{1}{2},0,3,-1,0,0,-\frac{1}{2},2,\frac{1}{2}\right)=-\frac{1}{2}p_0+3p_1^2-p_2-\frac{1}{2}p_1^4+2p_1^2p_2+\frac{1}{2}p_2^2 .
\end{align*}
More explicitly, we can write the irreducible characters of $SO(4)$ as explicit trace polynomials, 
\begin{equation*}
\left.\begin{array}{l}
\lambda_{0, 0}=0 \curly \chi_{0,0}(U)=4 \\
\lambda_{_{\frac{1}{2}, \frac{1}{2}}}=-\frac{3}{2} \curly  \chi_{_{\frac{1}{2},\frac{1}{2}}}(U)=2\tr(U) \\
\lambda_{1, 0}=-2 \curly  \chi_{1,0}(U)=\frac{1}{2}\tr^2(U)-\frac{1}{2}\tr(U^2) \\\lambda_{1, 1}=-4 \curly \chi_{1,1}(U) = -2+\tr^2(U)+\tr(U^2) \\
\lambda_{_{\frac{3}{2}, \frac{1}{2}}}=-\frac{9}{2} \curly  \chi_{_{\frac{3}{2},\frac{1}{2}}}(U)=-2\tr(U)+\frac{1}{2}\tr^3(U)-\frac{1}{2}\tr(U)\tr(U^2) \\
\lambda_{_{\frac{3}{2}, \frac{3}{2}}}=-\frac{15}{2} \curly  \chi_{_{\frac{3}{2},\frac{3}{2}}}(U)=2\tr(U)\tr(U^2)\\
\lambda_{2, 0}=-6 \curly  \chi_{2,0}(U)=-\frac{3}{2}\tr^2(U)-\frac{1}{2}\tr(U^2)+\frac{1}{4}\tr^4(U)-\frac{1}{2}\tr^2(U)\tr(U^2)+\frac{1}{4}\tr^2(U^2)\\
\lambda_{2, 1}=-8 \curly   \chi_{2,1}(U)=2-2\tr^2(U)+\frac{1}{4}\tr^4(U)-\frac{1}{4}\tr^2(U^2)\\
\lambda_{2, 2}=-12 \curly  \chi_{2,2}(U) = -2+3\tr^2(U)-\tr(U^2)-\frac{1}{2}\tr^4(U)+2\tr^2(U)\tr(U^2)+\frac{1}{2}\tr^2(U^2).\end{array}\right.
\end{equation*}

Due to the complexity of the general matrix $\left[\Delta_{SO(4)}\right]_{V_{\leq k}}$ written in the basis $\mathcal{B}_{\leq k}$, we cannot give an explicit formula for a general irreducible character $\chi_{j_1,j_2}$. But, one can compute them case by case as above, increasing the value of $k$.

\appendix

\section{Technical results}

\begin{lemma}\label{trace-comm}
    For $A,B\in \mathcal{M}_{N\times N}(\mathbb{R})$ we have
    $$\tr(K_{NN}\cdot (A\otimes B))=\tr(AB),$$
    where $K_{NN}$ is the $N^2\times N^2$ commutation matrix.
\end{lemma}
\begin{proof}
    For $i,j\in \{1,\dots,N\}$, the $N\times N$ block of the matrix $K_{NN}$ in the position $(i,j)$ is the single-entry matrix $J_{ji}$, which has 1 on the $(j,i)$ position and zeroes on the rest. For $i\in \{1,\dots,N\}$, the $(i,i)$ block matrix of the product $K_{NN}\cdot (A\otimes B)$ is $\sum\limits_{k=1}^N J_{ki}(a_{ki}B)$, therefore 
    $$\begin{aligned}
    \tr(K_{NN}\cdot (A\otimes B))&=\sum\limits_{i=1}^N\sum\limits_{k=1}^N a_{ki}\tr(J_{ki}B)
    =\sum\limits_{i,k,p,q=1}^N a_{ki}(J_{ki})_{pq}b_{qp}\\
    &=\sum\limits_{i,k,p,q=1}^N a_{ki}\delta_{kp}\delta_{iq}b_{qp}
    =\sum\limits_{i,k=1}^N a_{ki}b_{ik}=\tr(AB).
    \end{aligned}$$
\end{proof}

\begin{lemma}\label{lambda-comm}
    For $U\in \mathcal{M}_{N\times N}(\mathbb{R})$, we have
    $$\Lambda(U)\cdot K_{NN}=U^t\otimes U;~~K_{NN}\cdot \Lambda(U)=U\otimes U^t.$$
\end{lemma}
\begin{proof}
    For $i,j\in \{1,\dots,N\}$, the $(i,j)$ block of the product $\Lambda(U)\cdot K_{NN}$ is $$\sum\limits_{k=1}^N (\Lambda(U))_{ik}(K_{NN})_{kj}=\sum\limits_{k=1}^N (\bu_k\bu_i^t)J_{jk}.$$
    For $p,q\in \{1,\dots,N\}$, the $(p,q)$ element of the above $N\times N$ block is
    $$\sum_{k,r=1}^N (\bu_k\bu_i^t)_{pr}(J_{jk})_{rq}=\sum_{k,r=1}^N u_{pk}u_{ri}\delta_{jr}\delta_{kq}=u_{ji}u_{pq}=(U^t)_{ij}(U)_{pq},$$
    and therefore the first equality holds. The second equality follows by transposition.
\end{proof}

\begin{lemma}\label{xun} \emph{(\cite{xun}, \cite{magnus}) }
The Euclidean Hessian of $p_m(U)=\tr(U^m),~m\geq 2$, is given by
$$
\hs p_m(U)=m\cdot K_{NN}\left(\sum_{r=0}^{m-2}\left(U^t\right)^r \otimes U^{m-r-2}\right) .
$$
\end{lemma}


\begin{lemma}\label{produs-grad}
    For $f,g:SO(N)\to\mathbb{R}$, we have
    $$2\left<\nabla_{SO(N)}f(U),\nabla_{SO(N)}g(U)\right>=\tr((\nabla f)^t(U)\cdot \nabla g(U))-\tr(\nabla f(U)\cdot U^t\cdot\nabla g(U)\cdot U^t).$$
\end{lemma}
\begin{proof}
   For a smooth function $h:SO(N)\to \R$, in \cite{first-order} it has been proved the formula:
    \begin{equation}\label{gradient-ON}
    \nabla_{SO(N)}h(U)=\dfrac{1}{2}\left(\nabla h(U)-U\cdot(\nabla h)^t(U)\cdot U\right).
    \end{equation}
      Therefore,
    $$\begin{aligned}
     & 2 \left<\nabla_{SO(N)}f(U),  \nabla_{SO(N)}g(U)\right>= \\
     &=\dfrac{1}{2}\tr((\nabla f(U)-U\cdot(\nabla f)^t(U) \cdot U)^t(\nabla g(U)-U\cdot (\nabla g)^t(U)\cdot  U))\\
     &=\dfrac{1}{2}\tr((\nabla f)^t(U)-U^t\cdot (\nabla f)(U)\cdot U^t)(\nabla g(U)-U\cdot (\nabla g)^t(U)\cdot  U))\\
     &=\dfrac{1}{2}\left(\tr((\nabla f)^t(U)\cdot \nabla g(U))-\tr((\nabla f)^t(U)\cdot U\cdot (\nabla g)^t(U)\cdot U)\right.\\
     & ~~~\left.-\tr(U^t\cdot (\nabla f)(U)\cdot U^t\cdot (\nabla g)(U))+\tr(U^t\cdot (\nabla f)(U)\cdot U^t\cdot U\cdot (\nabla g)^t(U)\cdot U)\right)\\
     &=\tr((\nabla f)^t(U)\cdot \nabla g(U))-\tr(\nabla f(U)\cdot U^t\cdot\nabla g(U)\cdot U^t).
    \end{aligned}$$
\end{proof}

\begin{lemma}\label{interm}
\begin{enumerate}[(i)]
\item $\nabla_{SO(N)} p_1^q(U)=\dfrac{q}{2}p_1^{q-1}(U)(\mathbb{I}_N-U^2)$ for $q\geq 0$.
\item $\nabla_{SO(N)} p_m(U)=\dfrac{m}{2}\left((U^t)^{m-1}-U^{m+1}\right)$ for $m\geq 1$.
\end{enumerate}
\end{lemma}

\begin{proof}
{(i) The proof follows from \eqref{gradient-ON} and the equality $\nabla p_1^q=qp_1^{q-1}\mathbb{I}_N$.\\
(ii) The proof follows from \eqref{gradient-ON} and the equality $\nabla p_m=m(U^t)^{m-1}$ (see \cite{cookbook}).}
\end{proof}

\begin{lemma}\label{scalar-gradienti-76} For $m,m'\in \mathbb{N},~m \geq m'$ we have
$$2\left\langle\nabla_{SO(N)} p_m, \nabla_{SO(N)} p_{m'}\right\rangle=mm'\left(p_{m-m'}-p_{m+m'}\right).$$
\end{lemma}

\begin{proof}
{
From Lemma \ref{produs-grad} and $\nabla p_m=m(U^t)^{m-1},~ \nabla p_{m'}=m'(U^t)^{m'-1}$ we obtain
\begin{align*}
2\left\langle\nabla_{SO(N)} p_m, \nabla_{SO(N)} p_{m'}\right\rangle &=\tr((m(U^t)^{m-1})^t\cdot m'(U^t)^{m'-1})-\\
&-\tr(m(U^t)^{m-1}\cdot U^t\cdot m'(U^t)^{m'-1} \cdot U^t)\\
&=mm'(\tr(U^{m-1}\cdot (U^t)^{m'-1})-\tr((U^t)^{m+m'}))\\
&=mm'(\tr(U^{m-m'})-\tr(U^{m+m'}))\\
&=mm'\left(p_{m-m'}-p_{m+m'}\right).
\end{align*}}
\end{proof}

\begin{lemma}\label{scalar-product-gradient-234}
For  $p_{_{\boldsymbol{\lambda}'}}(U):=\tr(U^{m_1})\dots \tr(U^{m_r})$, where $m_1\geq \dots\geq m_r\geq 2$ and for $s>r$, we have 
$$\left<\nabla_{SO(N)}p_{\boldsymbol{\lambda}'},\nabla_{SO(N)}p_{1}^{s-r}\right>=\sum_{i=1}^r\dfrac{m_i(s-r)}{2}p_1^{s-r-1}p_{m_1}\dots\widehat{p_{m_i}}\dots p_{m_r}(p_{m_i-1}-p_{m_i+1}).$$
\end{lemma}

\begin{proof}
{
From Lemma \ref{interm} $(i)$ and $(ii)$, we obtain successively:
$$
\left<\nabla_{SO(N)}p_{\boldsymbol{\lambda}'},\nabla_{SO(N)}p_{1}^{s-r}\right>=$$
\begin{align*}
&=\left<\sum_{i=1}^r p_{m_1}\dots  \left(\nabla_{SO(N)}p_{m_i}\right)\dots p_{m_r},\dfrac{s-r}{2}p_1^{s-r-1}(\mathbb{I}_N-U^2)\right>\\
&=\left<\sum_{i=1}^r p_{m_1}\dots  \left(\frac{m_i}{2}((U^t)^{m_i-1}-U^{m_i+1})\right)\dots p_{m_r},\dfrac{s-r}{2}p_1^{s-r-1}(\mathbb{I}_N-U^2)\right>\\
&=\dfrac{s-r}{4}p_1^{s-r-1}\tr\left(\sum_{i=1}^r m_i\cdot p_{m_1}\dots  \left((U^t)^{m_i-1}-U^{m_i+1}\right)\dots p_{m_r}\cdot (\mathbb{I}_N-(U^t)^2)\right)\\
&= \dfrac{s-r}{4}p_1^{s-r-1}\sum_{i=1}^r m_i\cdot p_{m_1}\dots\widehat{p_{m_i}}\dots p_{m_r}\tr\left((U^t)^{m_i-1}-(U^t)^{m_i+1}-U^{m_i+1}+U^{m_i-1} \right)\\
&=\dfrac{s-r}{2}p_1^{s-r-1}\sum_{i=1}^r m_i\cdot p_{m_1}\dots\widehat{p_{m_i}}\dots p_{m_r}\left(p_{m_i-1}-p_{m_i+1} \right).
\end{align*}
}
\end{proof}

\noindent {\bf Proof of Lemma \ref{characters-in-V}}\\
From \cite{miller}, p. 351--362, we have the following formula for the characters of $SO(4)$
\begin{align*}
    \chi_{(i,j)}&=\frac{\sin (i+j+1)\frac{\alpha+\beta}{2}\cdot \sin (i-j+1)\frac{\alpha-\beta}{2}}{\sin \frac{\alpha+\beta}{2}\cdot \sin \frac{\alpha-\beta}{2}}\\
    &=U_{i+j}\left(\cos \frac{\alpha+\beta}{2}\right)\cdot U_{i-j}\left(\cos \frac{\alpha-\beta}{2}\right),
\end{align*}
where $i,j$ are integers such that $i\geq |j|$ and $U_m$ are the Chebyshev polynomials of second kind.
These polynomials have the explicit expression
$$U_m(x)=\sum_{s=0}^{[m/2]} (-1)^s\cdot C_{m-s}^s\cdot (2x)^{m-2s}.$$
If we make the following notations:
$$k:=i+j,~l:=i-j,~X:=\cos \frac{\alpha+\beta}{2},~Y:=\cos \frac{\alpha-\beta}{2},$$
then we have
\begin{equation*}
   \chi_{(i,j)}+\chi_{(i,-j)}=\sum_{q=0}^{[k/2]}\sum_{r=0}^{[l/2]}A(k,q)\cdot A(l,r)\cdot \left(X^{k-2q}Y^{l-2r}+X^{l-2r}Y^{k-2q}\right),
\end{equation*}
where $A(m,n):=(-1)^n\cdot C_{m-n}^n\cdot 2^{m-2n}$ for natural numbers $m,n$ such that $n\leq [m/2] $. We notice that $(k-2q)+(l-2r)=2(i-q-r)$ is an even integer.\\
Any expression of the form $X^aY^b+X^bY^a$, with $a,b$ natural numbers having {\bf the same parity}, can be written as
$$X^aY^b+X^bY^a=(XY)^{\min (a,b)}\left(X^{|a-b|}+Y^{|a-b|}\right)=(XY)^{\min (a,b)}\left((X^2)^{d}+(Y^2)^{d}\right),$$
where $d:=\frac{|a-b|}{2}$ is integer.\\\\
Since
$XY=\frac{\cos \alpha+\cos \beta}{2}=\frac{p_1}{4}$
and 
$X^2Y^2=\frac{p_1^2}{16},~~X^2+Y^2=1+\cos \alpha\cos \beta=\frac{p_1^2-p_2+4}{8},$
and taking into account that the expression $(X^2)^{d}+(Y^2)^{d}$ can be written as a symmetric polynomial in terms of the elementary symmetric functions $X^2+Y^2$ and $X^2Y^2$, we obtain that  $\chi_{(i,j)}+\chi_{(i,-j)}$ as an expression of trace polynomials $p_1$ and $p_2$. 

The irreducible characters of $SO(4)$ are $\chi_{(i,j)}+\chi_{(i,-j)}$, see \cite{biedenharn}. Making the link with our previous notation from Section \ref{SO(4)-123}, it follows that the irreducible character $\chi_{j_1,j_2}=\chi_{(i,j)}+\chi_{(i,-j)}$, where $j_1=\frac{i+j}{2}$ and $j_2=\frac{i-j}{2}$ verifying condition \eqref{conditie-j}, can be expressed as a polynomial in the variables $p_1$ and $p_2$. 

\rightline{$\Box$}

\end{document}